\documentclass[a4paper,english]{lipics-v2021}
\hideLIPIcs

\usepackage{microtype} 
\usepackage{todonotes}
\usepackage{xspace}
\usepackage{amssymb,amsfonts,amsmath}
\usepackage{todonotes}
\usepackage{enumerate}
\usepackage{soul}
\usepackage{lineno}
\usepackage{hyperref}
\usepackage[capitalise]{cleveref}
\usepackage[inline]{enumitem}
\usepackage{comment}
\usepackage{subcaption}
\usepackage[bibliography=common]{apxproof}

\newtheorem{prop}{Property}{\bfseries}{\itshape}
{\bfseries}{\itshape}
\newtheorem{obs}{Observation}{\bfseries}{\itshape}
\Crefname{theorem}{Theorem}{Theorems}
\Crefname{theorem2}{Theorem}{Theorems}
\Crefname{lemma}{Lemma}{Lemmas}
\Crefname{prop}{Property}{Properties}
\Crefname{obs}{Observation}{Observations}
\Crefname{question}{Question}{Questions}
\Crefname{invariant}{Invariant}{Invariants}
\Crefname{figure}{Fig.}{Figs.}

\newtheoremrep{theorem2}[theorem]{Theorem}
\newtheoremrep{corollary2}[theorem]{Corollary}
\newtheoremrep{lemma2}[theorem]{Lemma}
\newtheoremrep{claim2}[theorem]{Claim}
\newtheoremrep{property2}[theorem]{Property}
\newtheoremrep{proposition2}[theorem]{Proposition}
\newtheoremrep{observation2}[theorem]{Observation}

\nolinenumbers
\graphicspath{{figures/}}

\newcommand{\srac}[2]{\texorpdfstring{$#1$-bend $#2$-apRAC}{#1-bend #2-apRAC}}
\newcommand{\rac}[1]{\texorpdfstring{$#1$-bend apRAC}{#1-bend apRAC}}

\newcommand{\one}{I.\ref{inv:oblique-1}\xspace}
\newcommand{\two}{I.\ref{inv:unique}\xspace}
\newcommand{\three}{I.\ref{inv:oblique-2}\xspace}
\newcommand{\four}{I.\ref{inv:next-to}\xspace}
\newcommand{\dist}{\ensuremath{\operatorname{dist}}}

\title{Axis-Parallel Right Angle Crossing Graphs}

\author{Patrizio~Angelini}{John Cabot University, Rome, Italy}{pangelini@johncabot.edu}{0000-0002-7602-1524}{}

\author{Michael~A.~Bekos}{Department of Mathematics, University of Ioannina, Ioannina, Greece}{bekos@uoi.gr}{0000-0002-3414-7444}{}

\author{Julia~Katheder}{Wilhelm-Schickard-Institut f{\"u}r Informatik, Universit{\"a}t T{\"u}bingen, T{\"u}bingen, Germany}{julia.katheder@uni-tuebingen.de}{0000-0001-9186-3538}{}

\author{Michael~Kaufmann}{Wilhelm-Schickard-Institut f{\"u}r Informatik, Universit{\"a}t T{\"u}bingen, T{\"u}bingen, Germany}{mk@informatik.uni-tuebingen.de}{0000-0001-9186-3538}{}

\author{Maximilian~Pfister}{Wilhelm-Schickard-Institut f{\"u}r Informatik, Universit{\"a}t T{\"u}bingen, T{\"u}bingen, Germany}{maximilian.pfister@uni-tuebingen.de}{0000-0002-7203-0669}{}

\author{Torsten~Ueckerdt}{Institute of Theoretical Informatics, Karlsruhe Institute of Technology,Karlsruhe, Germany}{torsten.ueckerdt@kit.edu}{0000-0002-0645-9715}{}

\authorrunning{P.~Angelini, M.~A.~Bekos, J.~Katheder, M.~Kaufmann, M.~Pfister, T.~Ueckerdt}

\Copyright{Patrizio~Angelini, Michael A. Bekos, Julia~Katheder, Michael~Kaufmann, Maximilian~Pfister, Torsten~Ueckerdt}



\begin{document}



\maketitle

\begin{abstract}
A RAC graph is one admitting a RAC drawing, that is, a polyline drawing in which each crossing occurs at a right angle. Originally motivated by psychological studies on readability of graph layouts, RAC graphs form one of the most prominent graph classes in beyond planarity. 

In this work, we study a subclass of RAC graphs, called axis-parallel RAC (or apRAC, for short), that restricts the crossings to
pairs of axis-parallel edge-segments. apRAC drawings combine the readability of planar drawings with the clarity of (non-planar) orthogonal drawings. We consider these graphs both with and without bends. Our contribution is as follows: (i)~We study inclusion relationships between apRAC and traditional RAC graphs. (ii)~We establish bounds on the edge density of apRAC graphs. (iii)~We show that every graph with maximum degree $8$ is $2$-bend apRAC and give a linear time drawing algorithm. Some of our results on apRAC graphs also improve the state of the art for general RAC graphs. We conclude our work with a list of open questions and a discussion of a natural generalization of the apRAC model.
\end{abstract}
\newpage

\section{Introduction}

Planar graphs form a fundamental graph class in algorithms and graph theory. This is due to the fact that planar graphs have many useful properties, e.g., they are closed under minors and have a linear number of edges. Several decision problems, which are NP-complete for general graphs, become polynomial-time tractable, when restricted to planar inputs, e.g.~\cite{DBLP:journals/siamcomp/Hadlock75}. As a result, the corresponding literature is tremendously~large.

A recent attempt to extend this wide knowledge from planar to non-planar graphs was made in the context of \emph{beyond-planarity}, informally defined as a generalization of planarity encompassing several graph-families that are close-to-planar in some sense (e.g., by imposing structural restrictions on corresponding drawings). Notable examples are the classes of (i)~$k$-planar graphs~\cite{DBLP:journals/combinatorica/PachT97}, in which each edge cannot be crossed more than $k$ times, (ii)~$k$-quasi-planar graphs~\cite{DBLP:journals/combinatorica/AgarwalAPPS97}, which disallow $k$ pairwise crossing edges, and  (iii)~$k$-gap planar graphs~\cite{DBLP:journals/tcs/BaeBCEE0HKMRT18}, in which each crossing is assigned to one of the two involved edges such that each edge is assigned at most $k$ of its crossings. For an overview refer to the recent textbook~\cite{DBLP:books/sp/20/HT2020}. 

While all of the aforementioned graph-classes are topological, meaning that the actual geometry of the graph's elements is not important, there is a single class proposed in the literature that is purely geometric. The motivation for its study primarily stems from cognitive experiments indicating that the negative effect of edge crossings in a graph drawing tends to be eliminated when the angles formed at the edge crossings are large~\cite{DBLP:journals/vlc/HuangEH14}. In that aspect, the class of \emph{right-angle-crossing} (RAC) graphs forms the optimal case in this scenario, where all crossing angles occur at $90^\circ$. Formally, it was  introduced by Didimo, Eades and Liotta~\cite{DBLP:journals/tcs/DidimoEL11} a decade ago, and since then it has been a fruitful subject of intense research~\cite{DBLP:journals/jgaa/AngeliniCDFBKS11,DBLP:journals/algorithmica/GiacomoDEL14,DBLP:journals/mst/GiacomoDLM11,DBLP:journals/dam/EadesL13,DBLP:conf/esa/Forster020}.  

Generally speaking, the research on RAC graphs has focused on two main research directions depending on whether bends are allowed along the edges or not. Formally, in a $k$-bend RAC drawing of a graph each edge is a polyline with at most $k$ bends and the angle between any two crossing edge-segments is $90^\circ$. Accordingly, a $k$-bend RAC graph is one admitting such a drawing. A $0$-bend RAC graph (or simply RAC graph) with $n$ vertices has at most $4n-10$ edges~\cite{DBLP:journals/tcs/DidimoEL11}, that is, at most $n-4$ edges more than those of a corresponding maximal planar graph. The edge-density bounds for $1$- and $2$-bend RAC graphs are $5.5n-10$~\cite{DBLP:journals/tcs/AngeliniBFK20} and $74.2n$~\cite{DBLP:journals/comgeo/ArikushiFKMT12}, respectively, while for $k\geq 3$ it is known that every graph is $k$-bend RAC~\cite{DBLP:conf/esa/Forster020}. The research on RAC graphs, however, is not limited to edge-density bounds. 
Several algorithmic and combinatorial results~\cite{DBLP:journals/jgaa/AngeliniCDFBKS11,DBLP:journals/jgaa/ArgyriouBS12,DBLP:journals/jgaa/ArgyriouBKS13,DBLP:journals/cj/GiacomoDGLR15,DBLP:journals/ipl/DidimoEL10,DBLP:conf/esa/Forster020},~as well~as relationships with other graph classes~\cite{DBLP:journals/tcs/BekosDLMM17,DBLP:journals/tcs/BrandenburgDEKL16,DBLP:journals/comgeo/ChaplickLWZ19,DBLP:journals/ijcga/DehkordiE12,DBLP:journals/dam/EadesL13,DBLP:conf/swat/ChaplickFK020} are known; see~\cite{DBLP:books/sp/20/Didimo20}~for~a~survey.  

In this work, we continue the study of RAC graphs along a new and intriguing research line. Inspired by several well-established models for representing  graphs (including the widely-used orthogonal model~\cite{DBLP:journals/comgeo/BiedlK98,DBLP:conf/gd/FossmeierK95,DBLP:journals/siamcomp/GargT01}), we introduce and study a natural subfamily of $k$-bend RAC graphs, which restricts all edge segments involved in crossings to be axis-parallel. We call this class \rac{k}. We expect that this restriction will further enhance the readability of the obtained drawings, as these combine the simple nature of the planar drawings with the clarity of the (non-planar) orthogonal drawings by allowing non axis-parallel edge segments, only when those are crossing-free. We further expect that our restriction will lead to new results of algorithmic nature. As a matter of fact, almost all algorithms that have been already proposed in the literature about $k$-bend RAC graphs in fact yield \rac{k} drawings~\cite{DBLP:journals/jgaa/BekosDKW16,DBLP:journals/tcs/DidimoEL11,DBLP:conf/esa/Forster020}; e.g., every Hamiltonian degree-$3$ graph is \rac{0}~\cite{DBLP:journals/jgaa/ArgyriouBKS13}, while degree-$4$ and degree-$6$ graphs are $1$- and $2$-bend apRAC, respectively~\cite{DBLP:conf/mfcs/AngeliniBKKP22,DBLP:journals/jgaa/AngeliniCDFBKS11}. 

\noindent\medskip Our contribution is as follows:
\begin{itemize}
\item In \cref{sec:preliminaries} we study preliminary properties of \rac{0} graphs in order to prove that recognizing \rac{0} graphs is NP-hard (see \cref{thm:rac-0-np-hard}).

\item We study whether \rac{k} graphs form a proper subclass of  $k$-bend RAC graphs: For $k=0$, we establish a strict inclusion relationship with $K_6$ minus one edge being the smallest graph separating the two classes (see \cref{lem:different}). Further, our edge-density result for \rac{1} graphs establishes a strict inclusion relationship for $k=1$, see Corollary~\ref{cor:separatating-bend1}. The case $k=2$ is more challenging (due to the degrees of freedom introduced by bends) and we leave it as an open problem. For $k \geq 3$, the two classes coincide, as the construction establishing that every graph is $3$-bend RAC~\cite{DBLP:journals/tcs/DidimoEL11} can be converted to \rac{3} by a rotation of $45^\circ$.

\item We establish bounds on the edge density of $n$-vertex \rac{k} graphs: For $k=0$, we prove an upper bound of $4n-\sqrt{n}-6$ and give a corresponding lower bound construction with $4n-2\lfloor\sqrt{n}\rfloor - 7$ edges (see \cref{thm:aprac0-density}). For $k \in \{1,2\}$, we give linear upper bounds that are tight up to small additive constants (see \cref{thm:ap-1,thm:ap-2}). Notably, for $k=2$ our lower-bound construction is a graph with $n$ vertices and $10n-\mathcal{O}(1)$ edges. This bound extends to general $2$-bend RAC graphs and improves the previous best one of $7.83n - \mathcal{O}(\sqrt{n})$~\cite{DBLP:journals/comgeo/ArikushiFKMT12}, answering an open question in~\cite{DBLP:journals/tcs/AngeliniBFK20}.

\item We show that every graph with maximum degree $8$ is \rac{2} and give a linear time drawing algorithm (see \cref{thm:deg-8}) improving the previous best known result stating that $7$-edge colorable degree-$7$ graphs are $2$-bend (ap)RAC~\cite{DBLP:conf/mfcs/AngeliniBKKP22}.

\item Inspired by the slope-number of graphs, in \cref{sec:conclusions} we discuss a natural generalization of apRAC drawings where each edge segment involved in a crossing is parallel or perpendicular to a line having one out of $s$ different~slopes.
\end{itemize}

\section{Preliminaries}\label{sec:preliminaries}
Throughout this paper, basic graph drawing concepts are used as found in~\cite{DBLP:books/sp/20/HT2020,DBLP:reference/crc/2013gd}.
Let $G$ be a graph and $\Gamma$ be a polyline drawing of $G$ and let $e = (u,v)$ be an edge of $G$. We say that $e$ uses a horizontal (vertical) port at $u$ if the edge-segment of $e$ that is incident to $u$ is parallel to the $x$-axis (to the $y$-axis) in $\Gamma$. If $e$ uses neither a vertical nor a horizontal port at $u$, then it uses an \emph{oblique} port at $u$. In particular, we denote the four orthogonal ports (i.e., the vertical and the horizontal ports) as $\{N,E,S,W\}$-ports according to compass directions.
In a polyline drawing, vertices and bends are placed on grid-points, whereby the area of the drawing is determined by the smallest rectangular bounding box that contains the drawing.
In the following, we recall two properties that hold for $0$-bend RAC~drawings.
\begin{prop}[Didimo, Eades and Liotta~\cite{DBLP:journals/tcs/DidimoEL11}]\label{prop:no-fan}
In a $0$-bend RAC drawing no edge is crossed by two adjacent edges (see \cref{fig:prop1}). 
\end{prop}

\begin{figure}[t]
\centering
    \begin{subfigure}[b]{.22\textwidth}
    \centering    \includegraphics[scale=1,page=1]{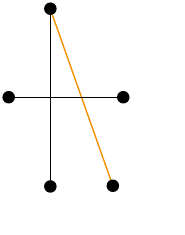}
    \subcaption{}
    \label{fig:prop1}
    \end{subfigure}
    \hfil
    \begin{subfigure}[b]{.22\textwidth}
    \centering    \includegraphics[scale=1,page=2]{figures/props1-4.pdf}
    \subcaption{}
    \label{fig:prop2}
    \end{subfigure}
    \hfil
    \begin{subfigure}[b]{.22\textwidth}
    \centering
\includegraphics[scale=1,page=3]{figures/props1-4.pdf}
    \subcaption{}
    \label{fig:prop3}
    \end{subfigure}
    \begin{subfigure}[b]{.24\textwidth}
    \centering    \includegraphics[scale=1,page=4]{figures/props1-4.pdf}
    \subcaption{}
    \label{fig:prop4}
    \end{subfigure}    
    \caption{Forbidden configurations by \cref{prop:no-fan,prop:no-triag,prop:no-three-outside,prop:vertex-in-triangle}.
    }
    \label{fig:props1-4}
\end{figure}

\begin{prop}[Didimo, Eades and Liotta~\cite{DBLP:journals/tcs/DidimoEL11}]\label{prop:no-triag}
A $0$-bend RAC drawing does not contain a triangle $T$ formed by edges of the graph and two edges $(u, v)$ and $(u, v')$, such that $u$ lies outside $T$ while both $v$ and $v'$ lie inside $T$ (see \cref{fig:prop2}).
\end{prop}

\noindent Next, we establish two properties limited to \rac{0} drawings.

\begin{prop}\label{prop:no-three-outside}
A \rac{0} drawing does not contain a triangle $T$ formed by edges of the graph and three vertices $v_1,v_2,v_3$ adjacent to a vertex $u$, such that $v_1,v_2,v_3$ lie outside $T$ and $u$ lies inside $T$ (see \cref{fig:prop3}).
\end{prop}
\begin{proof}
Assuming the contrary, \cref{prop:no-fan} implies that no two edges adjacent to $u$ cross the same boundary edge of $T$. Hence, $T$ consists of three axis-parallel edges; a contradiction.
\end{proof}

\begin{prop}\label{prop:vertex-in-triangle}
Let $\Gamma$ be a \rac{0} drawing containing a triangle $T$ formed by edges of the graph and two adjacent vertices $u$ and $v$ such that $u$ is contained inside $T$ while $v$ is outside $T$. Then, $\Gamma$ does not contain a vertex $w$ adjacent to $u$, $v$ and all vertices of $T$ (see \cref{fig:prop4}).
\end{prop}
\begin{proof}
For the sake of contradiction, assume there is a vertex $w$ adjacent to $u$, $v$ and all vertices of $T$.
If $w$ is inside $T$ in $\Gamma$, then $(v,u)$ and $(v,w)$ violate \cref{prop:no-triag}; a contradiction.
Otherwise, since $(u,v)$ and $(u,w)$ cross $T$, by \cref{prop:no-fan}, it follows that $T$ is a right-angled triangle whose legs are axis parallel. 
W.l.o.g., let $(v_1,v_2)$ and $(v_2,v_3)$ be the legs of $T$ crossed by $(u,v)$ and $(u,w)$, respectively, such that $(v_1,v_2)$ is horizontal and $(v_2,v_3)$ is vertical. 
It follows that the edge $(v_2,v_3)$ of $T$ is crossed by $(u,w)$ and $(w,v_1)$ violating \cref{prop:no-fan}; a contradiction.
\end{proof}
\noindent In Theorems \ref{thm:rac-0-np-hard} and \ref{thm:s-rac} we leverage the following property shown in~\cite{DBLP:journals/jgaa/ArgyriouBS12} of the so-called \emph{augmented square antiprism graph}. The gadget used in the NP-hardness proof of \cref{thm:rac-0-np-hard} is depicted in \cref{fig:np-d}, while the vertex-colored subgraph in \cref{fig:np-b} corresponds to the augmented square antiprism graph.

\begin{prop}[Argyriou, Bekos, Symvonis~\cite{DBLP:journals/jgaa/ArgyriouBS12}]\label{prop:antiprism-two-embbeddings}
Any straight-line RAC drawing of the augmented square antiprism graph
has two combinatorial embeddings.
\end{prop}

\section{0-bend apRAC graphs}\label{sec:aprac0}

In this section, we focus on properties of \rac{0} graphs. We start with an almost tight bound on the edge-density of \rac{0} graphs - for comparison, recall that $n$-vertex $0$-bend RAC graphs have at most $4n-10$ edges~\cite{DBLP:journals/tcs/DidimoEL11}.

\begin{theorem}\label{thm:aprac0-density}
A \rac{0} graph with $n$ vertices has at most $4n-\sqrt{n}-6$ edges. Also, there is an infinite family of graphs with $4n-2\lfloor\sqrt{n}\rfloor - 7$ edges that admit \rac{0} drawings.
\end{theorem}
\begin{proof}
For the upper bound consider any \rac{0} drawing $\Gamma$ of a graph~$G$ with $n$ vertices.
As a $(k \times k)$-grid has only $k^2$ grid points, we may assume without loss of generality that the vertices of $G$ use at least $\sqrt{n}$ different $y$-coordinates in $\Gamma$.
It follows that the subgraph $G_h$ of $G$ defined by the set $E_h$ of all horizontal edges of $\Gamma$ is a forest of paths with at least $\sqrt{n}$ components; at least one for each used $y$-coordinate.
Thus $|E_h| \leq n - \sqrt{n}$.
As $G - E_h$ is crossing-free in $\Gamma$, it has at most $3n-6$ edges, giving the desired upper bound of $4n - \sqrt{n}-6$ edges for $G$.

For the lower bound, consider the construction shown in \cref{fig:sl-lower}. For any even $k > 0$, construct a $k \times k$ grid graph $H_k$ which contains a pair of crossing edges in every quadrangular face. 
Let $G_k$ be the graph obtained from $H_k$ by adding two extremal adjacent vertices $N$ and $S$ connected to $2k-1$ consecutive boundary vertices of $H_k$ each (refer to the blue edges in \cref{fig:sl-lower} and observe that the edge between $N$ and $S$ can be added by moving $N$ upwards and to the right and $S$ downwards and to the right of $H_k$). If we denote by $n$ the number of vertices of $G_k$, then $n = k^2 + 2$, $k = \sqrt{n-2}$ and thus $m = 4n-2\lfloor \sqrt{n} \rfloor -7$.
\end{proof}
Since there exist $n$-vertex 0-bend RAC graphs with $4n-10$ edges, Corollary~\ref{cor:separatating} follows from \cref{thm:aprac0-density}. 

\begin{corollary}\label{cor:separatating}
The class of \rac{0} graphs is properly contained in the class of 0-bend RAC graphs.
\end{corollary}

\noindent In the following theorem, we show that $K_6$ minus one edge is the smallest graph that is $0$-bend RAC but not \rac{0}.

\begin{theorem}\label{lem:different}
The complete graph on $6$ vertices minus an edge $e$ is the minimal example separating the classes of $0$-bend RAC and \rac{0}.
\end{theorem}
\begin{proof}
Let $G = K_6-e$. \cref{fig:ap-rac-a} establishes that $G$ is RAC.
Suppose for a contradiction that $G$ admits a \rac{0} drawing $\Gamma$. Since $G$ has $6$ vertices and $14$ edges, it follows that $cr(G) \geq 2$. By the pigeonhole principle, there exists a vertex $u$ of $G$ which is incident to (at least) two crossing edges $(u,v)$ and $(u,w)$ in $\Gamma$. Denote by $e_1$ and $e_2$ the edges that cross $(u,v)$ and $(u,w)$, respectively. If $u$, $v$ and $w$ are colinear, the endpoints of $e_1$ and $e_2$ as well as vertices $u$,$v$ and $w$ are necessarily distinct, a contradiction since $G$ contains only six vertices. Otherwise, one of $\{(u,v),(u,w)\}$ is horizontal, while the other one is vertical; see \cref{fig:ap-rac-c}. Since $G$ is a complete graph minus one edge, at least one of the blue dotted edges is present in $\Gamma$ which is impossible by \cref{prop:no-fan}. Thus we conclude that $G$ is not \rac{0}.
To show the minimality of $G$, we first observe that every graph on five vertices is a subgraph of $K_5$, which admits a RAC drawing with exactly one crossing~\cite{DBLP:journals/tcs/DidimoEL11}, hence it is also a \rac{0} drawing after an appropriate rotation. To conclude, we provide \rac{0} drawings for the two non-isomorphic graphs which can be obtained by removing exactly two edges from $K_6$. In \cref{fig:ap-rac-d} two adjacent edges are removed, while in \cref{fig:ap-rac-e} two independent edges are removed from $K_6$.
\end{proof}

\begin{figure}[h]
    \centering
    \begin{subfigure}[b]{.24\textwidth}
    \centering
    \includegraphics[scale=0.7,page=1]{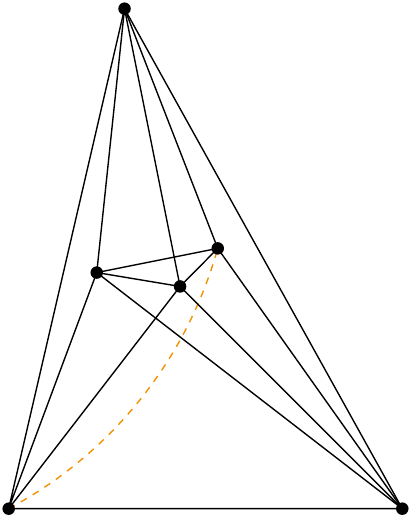}
    \subcaption{}
    \label{fig:ap-rac-a}
    \end{subfigure}
    \hfil
    \begin{subfigure}[b]{.24\textwidth}
    \centering
    \includegraphics[scale=1,page=2]{figures/K6-e.pdf}
    \subcaption{}
    \label{fig:ap-rac-c}
    \end{subfigure}
     \hfil
    \begin{subfigure}[b]{.24\textwidth}
    \centering
    \includegraphics[scale=0.7,page=4]{figures/K6-e.pdf}
    \subcaption{}
    \label{fig:ap-rac-d}
    \end{subfigure}
     \hfil
    \begin{subfigure}[b]{.24\textwidth}
    \centering
    \includegraphics[scale=0.7,page=5]{figures/K6-e.pdf}
    \subcaption{}
    \label{fig:ap-rac-e}
    \end{subfigure}
    \caption{
    Illustrations for the proof of \cref{lem:different}. Missing edges of $K_6$ are indicated in orange.}
    \label{fig:ap-rac}
\end{figure}

\noindent We conclude this section by studying the recognition problem of whether a graph is \rac{0}.

\begin{figure}[t]
\centering
    \begin{subfigure}[b]{.27\textwidth}
    \centering
    \includegraphics[scale=1,page=2]{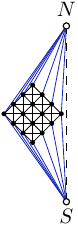}
    \subcaption{}
    \label{fig:sl-lower}
    \end{subfigure}
    \hfil
    \begin{subfigure}[b]{.51\textwidth}
    \centering
\includegraphics[scale=1,page=2]{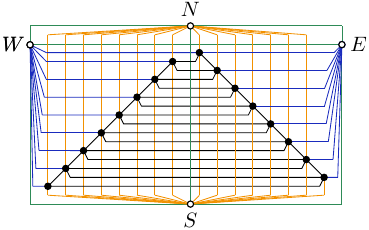}
    \subcaption{}
    \label{fig:aprac1-full}
    
    \end{subfigure}

    \caption{
    (a)~Lower bound construction for \rac{0}.
    (b)~Lower bound construction for \rac{1}.}
    \label{fig:ap0-ap1-lower}
\end{figure}

\begin{theorem}\label{thm:rac-0-np-hard}
It is NP-hard to decide whether a given graph is \rac{0}.
\end{theorem}
\begin{proof}
In order to prove the statement, we adjust the NP-hardness reduction (from 3-SAT) for the general case of straight-line RAC graphs introduced in \cite{DBLP:journals/jgaa/ArgyriouBS12}.
Based on the so-called \emph{augmented square antiprism graph}, which by \cref{prop:antiprism-two-embbeddings} has two combinatorial embeddings in the RAC setting, the construction of the clause-gadgets, the variable-gadgets as well as the connections between them is based on a basic \emph{building block} having the following properties: 
$(i)$ It has a unique embedding, 
$(ii)$ there are four vertices properly contained in its interior, which can be connected to vertices in its exterior by crossing a single boundary edge, 
$(iii)$ no edge can (completely) pass through it without forming a fan crossing and 
$(iv)$ it can be extended horizontally or vertically maintaining the aforementioned properties. 
Unfortunately, even though the augmented square antiprism graph is in fact \rac{0}, the building block of \cite{DBLP:journals/jgaa/ArgyriouBS12} is not.

\begin{figure}[ht]
    \centering
    \begin{subfigure}[b]{.32\textwidth}
    \centering
    \includegraphics[scale=1,page=3]{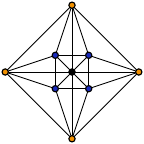}
    \subcaption{}
    \label{fig:np-b}
    \end{subfigure}
    \hfil
    \begin{subfigure}[b]{.32\textwidth}
    \centering
    \includegraphics[scale=1,page=2]{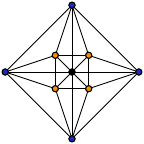}
    \subcaption{}
    \label{fig:np-c}
    \end{subfigure}
     \hfil
    \begin{subfigure}[b]{.32\textwidth}
    \centering
    \includegraphics[scale=1,page=4]{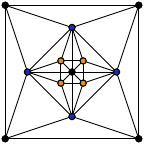}
    \subcaption{}
    \label{fig:np-d}
    \end{subfigure}
    \caption{
    Illustrations for the proof of \cref{thm:rac-0-np-hard}.}
    \label{fig:np}
\end{figure}

In the following, we prove that the graph $G$ of \cref{fig:np-d} satisfies properties $(i)-(iii)$; thus, it can act as the building block for our reduction. Observe that $G$ is composed of an augmented square antiprism graph $H$ and a $4$-cycle $C$ connected to the blue vertices of $H$. Recall that, by \cref{prop:antiprism-two-embbeddings}, $H$ has two combinatorial embeddings $\mathcal{E}_1$ and $\mathcal{E}_2$ in the 0-bend RAC setting and thus at most two in the \rac{0} setting; refer to \cref{fig:np-b,fig:np-c}, respectively. Since only axis-parallel edges are involved in crossings and since each crossing-free edge of $\mathcal{E}_1$ and $\mathcal{E}_2$ is not axis parallel, it follows that the crossing-free edges of $\mathcal{E}_1$ and $\mathcal{E}_2$ cannot be crossed in a \rac{0} drawing of $G$. 
This implies that the four blue vertices of $H$ which are connected to $C$ have to lie on a common face of the subgraph of $H$ induced by the crossing-free edges in $\mathcal{E}_1$ or $\mathcal{E}_2$. This is impossible in $\mathcal{E}_1$ and unique in  $\mathcal{E}_2$, as illustrated in \cref{fig:np-d}. This is enough to guarantee $(i)$.
Further, the figure clearly asserts that $(ii)$ and $(iii)$ are also guaranteed.
Property $(iv)$ can be guaranteed in the exact same way as in the original paper~\cite{DBLP:journals/jgaa/ArgyriouBS12}.

Since our building block is \rac{0} and since any crossing that does not involve a building block appears between axis-parallel edges in the original reduction, it follows that the constructed drawing is \rac{0} if and only if the input $3$-SAT formula is satisfiable.
\end{proof}

\section{1-bend apRAC graphs}\label{sec:aprac1}

In this section, we will establish an upper bound and an almost matching lower bound for the class of \rac{1} graphs. Recall that $n$-vertex $1$-bend RAC graphs have at most $5.5n-10$ edges~\cite{DBLP:journals/tcs/AngeliniBFK20}.
\begin{theorem}\label{thm:ap-1}
A \rac{1} graph with $n$ vertices has at most $5n-8$~edges. Also, there is an infinite family of graphs with $5n-16$ edges that admit \rac{1} drawings.
\end{theorem}
\begin{proof}
For the upper bound, consider a \rac{1} drawing $\Gamma$ of an $n$-vertex graph $G$. 
Each edge segment in $\Gamma$ is either horizontal (h), vertical (v) or oblique~(o). For $x,y \in \{h,v,o\}$, let $E_{xy}$ be the edges of $G$ with two edge segments of type $x$ and $y$. Then, $E_{hv}$, $E_{ho}$, $E_{vo}$ and $E_{oo}$ form a partition of the edge-set of $G$, assuming that edges that consist of only one $h$-, $v$- or $o$-segment are counted towards $E_{ho}$, $E_{vo}$ and $E_{oo}$, respectively.
By construction, any crossing involves exactly one vertical and one horizontal segment. Hence, the subgraph of $G$ induced by $E_{ho} \cup E_{oo}$ is planar and contains at most $3n-6$ edges.
Further, as every segment is incident to a vertex and since any vertex is incident to at most two vertical segments, we have 
$ |E_{vo} \cup E_{hv}| \leq 2n$.
We can assume that the topmost vertex $v_t$ is incident to at most one vertical edge-segment, since the edge segment incident to $v_t$ that points upwards cannot be involved in a crossing with a horizontal edge-segment. Otherwise, the endpoint incident to this edge segment would contradict the fact that $v_t$ is topmost in $\Gamma$. Hence, it can be replaced by a steep oblique edge-segment without introducing new crossings. Analogous observations can be made for the bottommost vertex in $\Gamma$, which implies that $ |E_{vo} \cup E_{hv}| \leq 2n-2$.
Thus, 
$|E| =|E_{ho}| + |E_{vo}| + |E_{hv}| + |E_{oo}| \leq 5n-8$.

Our lower bound construction is as follows; see \cref{fig:aprac1-full}. For $n \geq 7$, we arrange $n-4$ vertices forming a cycle  along the two legs of an isosceles triangle with a horizontal base (outer black edges), such that the left leg has $\lfloor \frac{n-4}{2}\rfloor$ vertices while the right one has $\lceil \frac{n-4}{2}\rceil$. These $n-4$ vertices are further joined by a $y$-monotone path of $n-7$ edges (inner black edges).
Two extremal vertices $N$ and $S$ above and below the triangle are connected to all $n-4$ vertices (orange edges). Similarly, two extremal vertices $W$ and $E$ to the left and right of the triangle are connected to all vertices of the left and right legs of the triangle respectively (blue edges); the topmost vertex of the right leg is also connected to $W$. Finally, we add six edges between the extremal vertices, which gives $n-4 + n-6 + 3(n-4)+6 = 5n-16$~edges.
\end{proof}

\noindent Since there exists $1$-bend RAC graphs with $5.5n-72$ edges~\cite{DBLP:journals/tcs/AngeliniBFK20}, the following corollary is immediate.
\begin{corollary}\label{cor:separatating-bend1}
The class of \rac{1} graphs is a proper subclass of the one of $1$-bend RAC graphs.
\end{corollary}

\section{2-bend apRAC graphs}\label{sec:aprac2}
In \cref{thm:ap-2}, we provide an upper-bound for the edge density of \rac{2} graphs together with a lower-bound construction which is tight up to an additive constant. Our result provides a stark contrast to the one for $2$-bend RAC graphs, where the current best upper-bound on the number of edges of $n$-vertex graphs is $74.2n$~\cite{DBLP:journals/comgeo/ArikushiFKMT12}, while the previous best lower bound-construction contained only $7.83n - \mathcal{O}(\sqrt{n})$~\cite{DBLP:journals/comgeo/ArikushiFKMT12} edges.

\begin{theorem}\label{thm:ap-2}
A \rac{2} graph with $n$ vertices has at most $10n-12$~edges. Also, there is an infinite family of graphs with $10n - 46$ edges that admit \rac{2} drawings.
\end{theorem}
\begin{proof}
Consider a \rac{2} drawing $\Gamma$ of an $n$-vertex graph $G$. 
Each edge segment in $\Gamma$ is either horizontal (h), vertical (v) or oblique (o).
Denote by $S$ the set of edges that contain at least one segment in $\{h,v\}$ incident to a vertex. Since any vertex is incident to at most two vertical and at most two horizontal segments, it follows that $|S| \leq 4n$. 
Let $E_h$, $E_v$ and $E_o$ be the set of edges of $E \setminus S$ whose middle part is $h$, $v$ and $o$, respectively. Assuming that an edge of $E \setminus S$ consisting of less than three segments belongs to $E_o$, it follows that $E_h$, $E_v$ and $E_o$ form a partition of $E \setminus S$.
Observe that the edges of $E_o$ cannot be involved in any crossing in $\Gamma$, as all of its segments are oblique. Further, no two edges of $E_h$ or of $E_v$ can cross.
Hence, the subgraphs induced by $E_h \cup E_o$ and $E_v \cup E_o$ are planar and contain at most $3n-6$ edges each.~Recall~that $|S| \leq 4n$ and thus $|E| \leq |S| +|E_h| + |E_v| + 2|E_o| \leq 4n + 3n-6 +3n-6 = 10n-12$.

Refer to \cref{fig:10n-construction} for a schematization of the upcoming lower-bound construction and to \cref{fig:10n-full-construction} for a concrete example.
Fix an integer $k \geq 6$ and consider a set $P$ of $k^2$ points of a $k \times k$ square grid in the plane but rotated very slightly, say counterclockwise, so that the points in each column have consecutive $x$-coordinates (consequently the points in each row have consecutive $y$-coordinates).
For two points $p,q \in P$ let their \emph{$x$-distance} $\dist_x(p,q)$ be the number of points in $P$ having their $x$-coordinate between $p$ and $q$.
Similarly define the $y$-distance $\dist_y(p,q)$.
The crucial property of point set $P$ is the~following.
\begin{equation}
    \text{For any $p \neq q \in P$ we have $\dist_x(p,q) + \dist_y(p,q) \geq k-1 \geq 5$.}\label{eq:S-property}
\end{equation}
 
\begin{figure}[t]
    \begin{subfigure}[b]{.48\textwidth}
        \centering
        \includegraphics[page=4,scale=0.8]{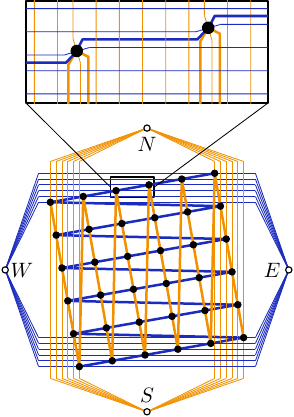}
        \subcaption{}
        \label{fig:10n-construction}
    \end{subfigure}
    \begin{subfigure}[b]{.48\textwidth}
        \centering
        \includegraphics[scale=0.6,page=6]{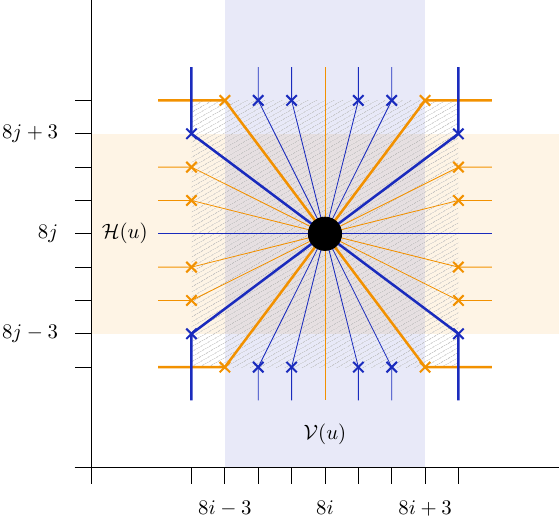}
        \subcaption{}
        \label{fig:the-box}
    \end{subfigure}    
    \caption{(a)~Illustration of the construction in \cref{thm:ap-2} with $k=6$. 
    Edges with two oblique segments are indicated in orange for vertical middle segments and in blue for horizontal middle segments.
    Edges using the horizontal or vertical ports are omitted for readability.
    (b)~Edge routing in the $8 \times 8$ box $B(u)$ of a vertex $u$. Blue ports are exclusively used by edges of $F_1$ and $F_3$ and orange ports by $F_2$ and $F_4$. Note that the ports illustrated by bold lines are reserved for oblique-$2$ edges. Bends on the border of the box are emphasized by a cross.
    }
\end{figure}

Between any pair $p,q \in P$ with consecutive $x$-coordinates, i.e., $\dist_x(p,q) = 0$, we add a $2$-bend edge with vertical middle segment by starting and ending with a very short oblique segment at $p$ respectively $q$.
Similarly, we add a $2$-bend edge with horizontal middle segment when $\dist_y(p,q) = 0$.
Note that these are in total $2k^2-2$ edges, no two of which connect the same pair of points, due to \eqref{eq:S-property}. 

Next we add four additional points $N,E,S,W$ to the top, right, bottom, and left of all points in $P$, respectively.
For every point $p$ we add a $2$-bend edge with vertical middle segment between $p$ and $N$ starting with a very short oblique segment at $p$ and ending with an almost horizontal (but still oblique) segment at $N$.
Similarly, we add a $2$-bend edge with vertical middle segment between $p$ and $S$, as well as one with horizontal middle segment to each of $E,W$.
Note that these are in total $4k^2$ edges, and that all oblique segments can be chosen such that all crossings involve middle segments only.

Next we add for (almost) each point $p \in P$ four more $2$-bend edges.
First, consider for $p$ the point $q \in P$ to the right of $p$ with $\dist_x(p,q) = 1$, unless $p$ is one of the two rightmost points in $P$.

We draw a $2$-bend edge from $p$ to $q$ by starting with a horizontal segment at $p$ to almost the $x$-coordinate of $q$, continuing with a vertical segment to almost the $y$-coordinate of $q$, and ending with a very short oblique segment at $q$.
Similarly, we use the left horizontal port at $p$ for an edge to the point $q$ left of $p$ with $\dist_x(p,q)=2$.
(We take $x$-distance $2$ instead of $1$ to avoid introducing a parallel edge.)
Symmetrically, we draw two edges using the vertical ports at $p$.
Note that these are in total $4k^2 - 10$ edges, and that all crossings involve horizontal and vertical segments only.

Finally, we add easily add six edges to create a $K_4$ on vertices $N,E,S,W$.
To conclude, we have constructed a \rac{2} graph with $n = k^2+4$ vertices and $
 (2k^2 - 2) + 4k^2 + (4k^2 - 10) + 6 = 10k^2 - 6 = 10n - 46$ edges.\qedhere

\begin{figure}[t!]
    \centering
    \includegraphics[width=.97\textwidth,page=3]{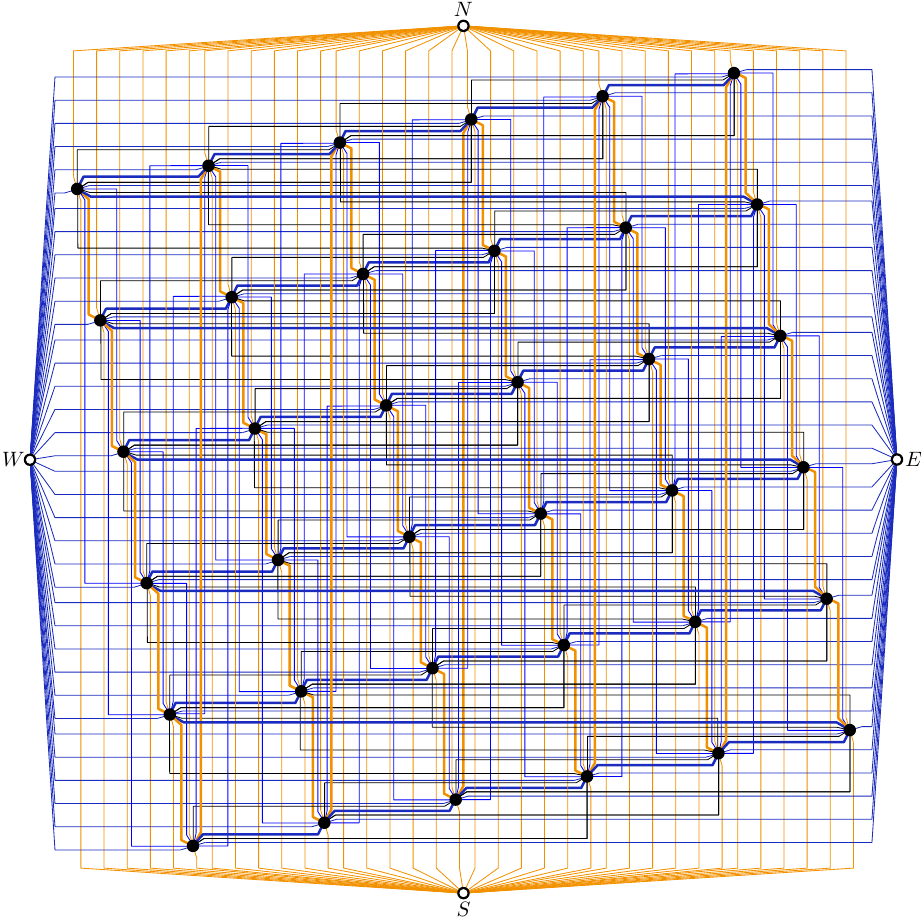}
    \caption{Illustration of the construction in \cref{thm:ap-2} with $k=6$. The $K_4$ on the vertices $N,E,S,W$ is omitted due to space reasons.}
    \label{fig:10n-full-construction}
\end{figure}    
\end{proof}

\section{Every graph with maximum degree 8 is 2-bend apRAC} \label{sec:degree-8}
In the following, we prove that graphs with maximum degree $8$ admit \rac{2} drawings of quadratic area which can be computed in linear time.
We leverage the following result in order to decompose the input graph.

\begin{lemma}
[Eades, Symvonis, Whitesides~\cite{DBLP:journals/dam/EadesSW00}] \label{thm:2-factors} Let $G=(V,E)$ be an n-vertex undirected graph of degree $\Delta$ and let $d = \lceil \Delta / 2 \rceil$. Then, there exists a directed multigraph $G' = (V,E')$ with the following properties:
\begin{enumerate}
    \item each vertex of $G'$ has indegree $d$ and outdegree $d$;
    \item $G$ is a subgraph of the underlying undirected graph of $G'$; and
    \item the edges of $G'$ can be partitioned into $d$ edge-disjoint directed $2$-factors (where a $2$-factor is a spanning subgraph of $G'$ consisting of vertex disjoint cycles, called cycle cover in~\cite{DBLP:journals/dam/EadesSW00}).
\end{enumerate}

\noindent The directed graph $G'$ and the $d$ $2$-factors can be computed in $\mathcal{O}(\Delta^2 n)$ time.
\end{lemma}

\noindent Now, we are ready to state the main result.

\begin{theorem} \label{thm:deg-8}
Given a graph $G$ with maximum degree $8$ and $n$ vertices, it is possible to compute in $\mathcal{O}(n)$ time a \rac{2} drawing of $G$ with $\mathcal{O}(n^2)$ area.
\end{theorem}

\begin{proof}
Let $G$ be a simple graph with maximum degree $8$ and $n$ vertices. We apply Lemma~\ref{thm:2-factors} to augment $G$ to a \textbf{directed} $8$-regular multigraph having four edge-disjoint $2$-factors $F_1$, $F_2$, $F_3$ and $F_4$. Before we present our algorithm in full detail, we sketch an outline of the necessary steps. We want to stress that in the following, the direction of an edge $(u,v)$ plays an important role and hence we consider it as a directed edge with source $u$ and target $v$. 
\subsection{Outline of the algorithm}
In the first step, we will construct two total orders $\prec_x$ and $\prec_y$ of the vertices of $G$ which will determine the $x$- and $y$-coordinates of the vertices in the final drawing. In particular, if vertex $u$ of $G$ has the $i$-th position in $\prec_x$ and the $j$-th position in $\prec_y$, then $u$ will be placed at point $(8i,8j)$ in the final drawing.
We will construct these two orders independently such that $\prec_x$ is defined by $F_1 \cup F_3$ and $\prec_y$ is defined by $F_2 \cup F_4$. 
After the computation of $\prec_x$ and $\prec_y$, which finalizes the position of the vertices in our resulting drawing $\Gamma$, it remains to draw the edges which are fully characterized by the placement of the respective bend-points. Every edge will be drawn with exactly three segments, which are either horizontal, vertical or oblique.
To ensure that all crossings in $\Gamma$ occur between horizontal and vertical segments, we will restrict oblique segments to be ``short'' (a precise definition follows below) and require that they are incident to a vertex. To this end, we will define, for each vertex $u$ of $G$, a closed box $B(u)$ centered at $u$ of size $8 \times 8$, such that the oblique segments incident to $u$ are fully contained inside $B(u)$. Note that by construction, the interior of two boxes do not overlap (they may touch at a corner).
Since the $x$-coordinate of two consecutive vertices $u$ and $v$ of $\prec_x$ differs by exactly $8$, there is a vertical line that is (partially) contained inside both $B(u)$ and $B(v)$ (analogous for a horizontal line and consecutive vertices in $\prec_y$). This allows us to join $u$ and $v$ by an edge that consists of two oblique segments, which is called an \emph{oblique-$2$} edge. If the unique orthogonal segment of an oblique-$2$ edge is vertical (horizontal), we will refer to it as a vertical (horizontal) oblique-$2$ edge. An edge that contains exactly one oblique segment will analogously be called an \emph{oblique-$1$} edge.

In the second step, we will classify every edge of $G$ as either an oblique-$1$ or an oblique-$2$ edge - again this classification is done independently for $F_1 \cup F_3$ and $F_2 \cup F_4$; we focus on the description of $F_1 \cup F_3$, the other one is symmetric. Let $e = (u,v)$ be an edge of $F_1 \cup F_3$. If $u$ and $v$ are consecutive in $\prec_x$, then $e$ is classified as a vertical oblique-$2$ edge. Otherwise, $e$ is classified as an oblique-$1$ edge such that the (unique) oblique segment is incident to the target $v$, while the orthogonal segment at $u$ uses the $E$-port at $u$ if $u \prec_x v$, otherwise it uses the $W$-port. 

In the final step, we will specify the exact coordinates of the bend-points. At a high level, oblique segments (which are by construction all incident to vertices) will end at the boundary of the corresponding box, see \cref{fig:the-box}.
The bend-points between vertical and horizontal segments are then naturally defined by the intersections of their corresponding lines. 

The final drawing $\Gamma$ will then satisfy the following two properties.
\begin{enumerate*}[label={(\roman*)}, ref=(\roman*)]
\item\label{prp:no-edge-overlap}No bend-point of an edge lies on another edge and
\item\label{prp:in-the-box}the edges are drawn with two bends each so that only the edge segments that are incident to $u$ are contained in the interior of $B(u)$, while all the other edge segments are either vertical or horizontal. 
\end{enumerate*}
This will guarantee that the resulting drawing is $2$-bend RAC; for an example see Fig.~\ref{fig:k9}. Note that \ref{prp:no-edge-overlap} guarantees that no two segments have a non-degenerate overlap.

\subsection{Computing $\prec_x$ and $\prec_y$}
\label{subsec:compute_orders}
We will now describe how to construct $\prec_x$ and $\prec_y$ explicitly.
We focus on the construction of $\prec_x$ which is based on $F_1$ and $F_3$, the order $\prec_y$ can be constructed analogously.
Let $C_1,C_2,\dots,C_k$ be an arbitrary ordering of the components of $F_1$. Recall that by definition, each such $C_i$ is a directed cycle. Let $S$ be a set of vertices that contains exactly one arbitrary vertex from each cycle in $F_1$ and let $P_1,P_2,\dots,P_k$ be the resulting directed paths obtained by restricting the cycles to $V \setminus S$. Note that this may yield paths that are empty, i.e., when the corresponding cycle consists of a single vertex. 
We construct $\prec_x$ (limited to $V \setminus S$) such that the vertices of each path appear consecutively defined by the unique directed walk from one endpoint of the path to the other. 
 The relative order between paths is $P_1 \prec_x P_2 \prec_x \dots \prec_x P_k$.
Hence it remains to insert the vertices of $S$ into $\prec_x$. Throughout the algorithm, we will maintain the following invariant which will ensure the correctness of our approach.
\begin{enumerate}[label={\bf I.\arabic*}, ref={\arabic*}]
    \item \label{inv:next-to} Let $u \in S$ be a vertex of cycle $C_i$. If $|C_i| > 1$, then $u$ is placed next to at least one vertex of $P_i$. Otherwise, $u$ is placed directly after the last vertex of $C_{i-1}$ (or as first vertex if $i = 1$) in $\prec_x$.
\end{enumerate}

\noindent If \four is maintained, we can guarantee the following observation.

\begin{obs}\label{obs:comparable}
Let $u \in C_i$ and $v \in C_j$ be two vertices of $G$ with $i \neq j$. Then, the relative order of $u$ and $v$ in $\prec_x$ is known.
\end{obs}

Assume that each vertex in $S$ that belongs to $C_1,\dots,C_{i-1}$ has been inserted in $\prec_x$.
Let $u \in S$ be the vertex that belongs to $C_i \setminus P_i$. If $|C_i| \leq 2$, then we place $u$ immediately after the last vertex of $C_{i-1}$ in $\prec_x$ if $i > 1$, otherwise $u$ is the first vertex of $\prec_x$ which maintains \four.
Hence, in the remainder we can assume that $C_i$ consists of at least three vertices. Let $a$, $b$ and $c$ be the vertices of $G$ such that $ (u,a)$, $(b, u) \in F_1$ and $(u,c) \in F_3$. Even though $G$ is a multigraph, we have that $a \neq b$ since $C_i$ contains at least three vertices. 
Hence, by construction we have $a \prec_x b$ - in particular, $a$ is the first vertex of $P_i$ in $\prec_x$, while $b$ is the last one. 
Let $C_j$ (possibly $j =i$) be the cycle that contains $c$. Note that it is possible that $c \in S$, i.e., $c$ is not part of $\prec_x$ initially. However, as this can only happen if $i \neq j$, we know the relative position of $u$ and $c$ by Observation~\ref{obs:comparable}.
We distinguish between the following cases based on the relative order of cycle $C_i$ (which contains $u$) and cycle $C_j$ (which contains $c$) in $\prec_x$.
\begin{enumerate}
\label{deg-8-v-3}
    \item $j < i$. 
    \label{deg-8-case-1}
     We insert $u$ immediately before $a$ in $\prec_x$ such that it is the first vertex of $C_i$, see \cref{fig:deg-8-precx-1}. Clearly, this maintains \four. 
     
    \item $i < j$.
     \label{deg-8-case-2}
     This case is symmetric to the previous one - we insert $u$ immediately after $b$ in $\prec_x$ such that it is the last vertex of $C_i$, see \cref{fig:deg-8-precx-2}, which again maintains \four.  

    \begin{figure}[t]
\centering
    \includegraphics[scale=0.6,page=7]{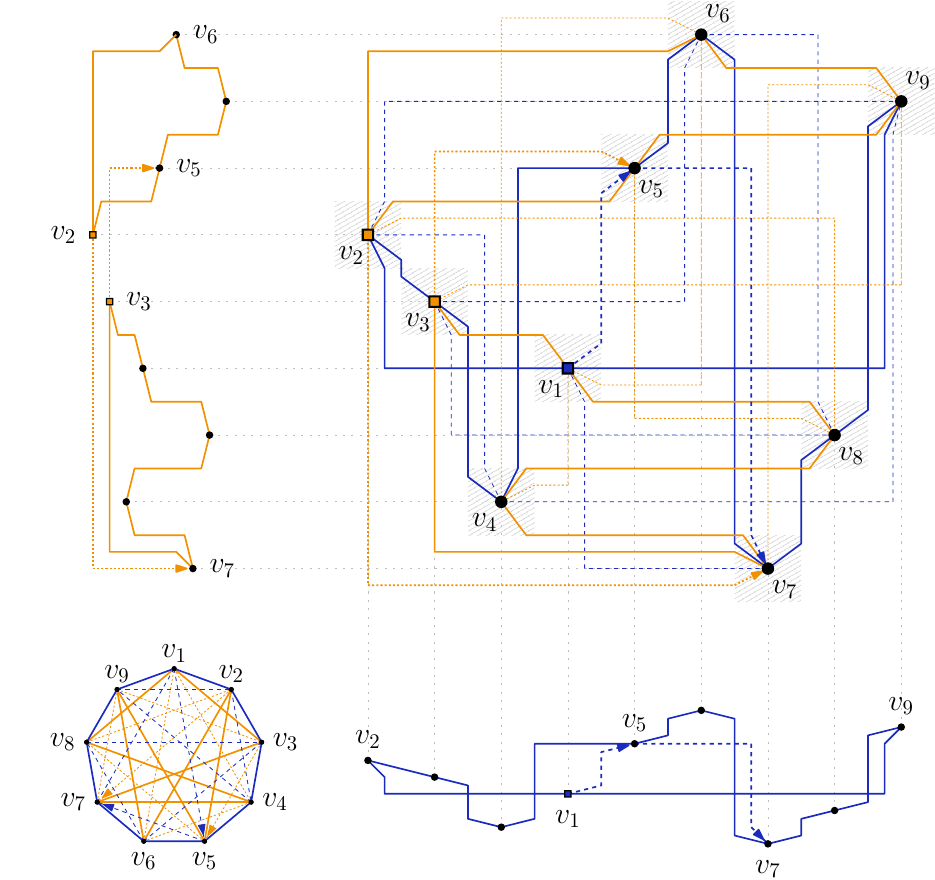}
    \caption{A \rac{2} drawing of $K_9$; 
    $F_1$ and $F_3$ are blue; 
    $F_2$ and $F_4$ are orange. 
    Below the drawing of $K_9$ there is a illustration of the cycles in $F_1$ and the relevant edges in $F_3$ for positioning $v_1 \in S$ according to Case~\ref{deg-8-case-3} in the construction of $\prec_x$. Similarly, a visualization of the cycles in $F_2$ and the relevant edges in $F_4$ is displayed to the left.
    }
    \label{fig:k9}
\end{figure}
    \item $i = j$.
    \label{deg-8-case-3}
    In this case, we have that $c$ also belongs to $C_i$ (in particular, $c$ belongs to $P_i$ and thus is already part of $\prec_x$). 
    If $c = a$ or $c = b$, we simply omit the edge $(u,c)$ and proceed as in the first case, i.e., we place $u$ as the first vertex of $C_i$.
    Otherwise, we insert $u$ directly before or directly after $c$ in $\prec_x$ based on the edge $(c, d) \in F_3$. The relative order of $c$ and $d$ in $\prec_x$ is known by Observation~\ref{obs:comparable} unless $d \in C_i$. If $d \in P_i$, the relative order between $c$ and $d$ is also known (as both are already present in $\prec_x$). If $d \notin P_i$, then $d = u$ and we can omit the edge $(u,c) \in F_3$ (because it is a copy of $(c,d) \in F_3$), in which case we can again proceed as in the first case. Hence, $d \neq u$ holds.
    If $c \prec_x d$, we insert $u$ directly before $c$ in $\prec_x$, see see \cref{fig:deg-8-precx-3}, otherwise we insert $u$ directly after $c$ in $\prec_x$. In both cases, we maintain \four. 
\end{enumerate}

\begin{figure}[t]
    \begin{subfigure}[b]{.31\textwidth}
    \centering
    \includegraphics[width=\textwidth,page=19]{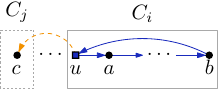}
    \subcaption{}
    \label{fig:deg-8-precx-1}
    \end{subfigure}
    \hfil
    \begin{subfigure}[b]{.31\textwidth}
    \centering
    \includegraphics[width=\textwidth,page=20]{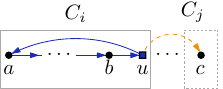}
    \subcaption{}
    \label{fig:deg-8-precx-2}
    \end{subfigure}
    \hfil
    \begin{subfigure}[b]{.29\textwidth}
    \centering
    \includegraphics[width=\textwidth,page=21]{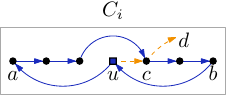}
    \subcaption{}
    \label{fig:deg-8-precx-3}
    \end{subfigure}
    \label{fig:deg-8-cases}  
    \caption{Illustration of the construction of $\prec_x$, Case~\ref{deg-8-case-1} is shown in (a), Case~\ref{deg-8-case-2} in (b) and Case~\ref{deg-8-case-3} in (c). Blue edges belong to $F_1$, while dashed orange edges belong to $F_3$.}
\end{figure}

This concludes our construction of $\prec_x$.

\subsection{Classification of the edges and port assignment}
We focus on the classification of the edges of $F_1 \cup F_3$ and their port assignment, the classification of the edges of $F_2 \cup F_4$ is analogous.
Our classification will maintain the following invariants.
\begin{enumerate}[label={\bf I.\arabic*}, ref={\arabic*}]
\setcounter{enumi}{1}
    \item \label{inv:oblique-2} The endpoints of each vertical oblique-$2$ edge are consecutive in~$\prec_x$.
    \item \label{inv:oblique-1} Each oblique-$1$ edge $(u,v) \in F_1 \cup F_3$ is assigned the $W$-port at its source vertex $u$, if $v \prec_x u$; otherwise, if $u \prec_x v$, it is assigned the $E$-port at $u$.
    \item \label{inv:unique} Every horizontal port is assigned at most once.
\end{enumerate}

Let us consider an edge $e \in F_1 \cup F_3$ between vertices $u$ and $v$. If $u$ and $v$ are consecutive in $\prec_x$, then we classify $e$ as a vertical oblique-$2$ edge. If $u$ and $v$ are not consecutive in $\prec_x$, we will classify $e$ as an oblique-$1$ edge, which therefore guarantees \three. For any oblique-$1$ edge, we will, in an initial phase, assign the ports precisely as stated in \one. In a subsequent step, we will create a unique assignment of the horizontal ports by reorienting some edges of $F_1 \cup F_3$ in order to guarantee \two.
Suppose that after the initial assignment, there exists a vertex $u$ such that one of its orthogonal ports is assigned to two oblique-$1$ edges. Assume first the $W$-port of $u$ is assigned to edges $(u,a)$ and $(u,b)$.
By construction, $u$ has exactly one outgoing edge in $F_1$, say $(u,a)$, and exactly one outgoing edge in $F_3$, say $(u,b)$.
Let $C_i$ be the cycle of $F_1$ that contains both $u$ and $a$ (which implies that $|C_i| > 1$, as we omit self-loops) and let $C_j$ be the cycle that contains $b$ (possibly $i=j$). Recall that by construction, the vertices of $P_i$ appear consecutively in $\prec_x$ before the insertion of the vertex $v \in C_i \setminus P_i$. Since $(u,a)$ is an oblique-$1$ edge, we have that $u$ and $a$ are not consecutive in $\prec_x$. If $|C_i|= 2$, one of $u$ or $a$ coincides with $v$, but then $u$ and $a$ are consecutive in $\prec_x$ and thus the edge $(u,a)$ is an oblique-$2$ edge. Hence, $|C_i| > 2$ holds and we either have $u = v$, $a = v$ or $v$ was inserted directly in between $a$ and $u$. In the following, we will refer to Cases~\ref{deg-8-case-1} - \ref{deg-8-case-3} of \cref{subsec:compute_orders}, where we computed the total order $\prec_x$. 

\begin{enumerate}
	\item $u = v$. Assume first that $C_i \neq C_j$. Then, since $(u,b)$ is assigned the $W$-port at $u$, we have $b \prec_x u$ by \one which implies $j < i$ and hence we placed $u$ according to Case~\ref{deg-8-case-1}, i.e., as the first vertex of $C_i$ in $\prec_x$. But since $a \in C_i$, we then have $u \prec_x a$ and thus $(u,a)$ would use the $E$-port at $u$, a contradiction. 
	
	Hence assume that $C_i = C_j$, i.e., $b \in C_i$. Then we are in Case~\ref{deg-8-case-3}. In particular, we placed $u$ such that $u$ and $b$ are consecutive, thus $(u,b)$ is classified as an oblique-$2$ edge, again we obtain a contradiction.
	
	\item $a = v$. Since  $(u,a)$ uses the $W$-port at $u$ by assumption, we have that $a \prec_x u$ by \one and thus $a$ cannot be the last vertex of $C_i$ in $\prec_x$, and so we are in Case~\ref{deg-8-case-1}~or~\ref{deg-8-case-3}. In Case~\ref{deg-8-case-1}, $a$ is placed as the first vertex of $C_i$ since there exists a vertex $a'$ with $a' \prec_x a$ such that $(a,a') \in F_3$. Further, $a$ is placed next to vertex $v'$ (i.e., the first vertex of $P_i$ in $\prec_x$) with $(a,v') \in F_1$ by construction. Then, we can redirect the edge $(u,a) \in F_1$ such that we can assign $(a,u)$ the $E$-port at $a$ which solves the conflict at $u$ and does not introduce a conflict at $a$ which guarantees \two.
	In Case~\ref{deg-8-case-3}, $a$ was placed consecutive to vertex $a' \in P_i$ with $(a,a') \in F_3$. As $(u,a)$ uses the $W$-port at $u$, $u$ is necessarily the last vertex of $C_i$ in $\prec_x$. Since the other neighbor of $a$ in $F_1$ different from $u$ is the first vertex of $C_i$ in $\prec_x$, i.e., it precedes $a$ in $\prec_x$, we can again reorient the edge $(u,a)$ and assign the edge $(a,u)$ to the free $E$-port of $a$, solving the conflict at $u$ which guarantees \two.
	
	\item $v$ was inserted directly in between $a$ and $u$. In this case, we have that both $a$ and $u$ belong to $P_i$. Since we assume that $(u,a)$ uses the $W$-port at $u$, it follows that $a \prec_x u$ holds. But then by construction, the edge of $F_1$ that joins $a$ und $u$ is directed from $a$ to $u$ and we obtain a contradiction.
\end{enumerate}

The case where the $E$-port of $u$ is assigned to two edges can be solved in a similar way. 
More precisely, assume  that the $E$-port of $u$ is assigned to edges $(u,a)$ and $(u,b)$ with $(u,a) \in F_1$ and $(u,b) \in F_3$.
The following case analysis might look symmetric to the one where the $E$-port of $u$ is assigned twice, however there exist some subtle differences, in particular in the second case.
Again, let $C_i$ be the cycle of $F_1$ that contains both $u$ and $a$ (which implies that $|C_i| > 1$) and let $C_j$ be the cycle that contains $b$ (possibly $i=j$). Since $u$ and $a$ are not consecutive in the final $\prec_x$, we either have $u = v$, $a = v$ or $v$ was inserted directly in between $u$ and $a$ which gives rise to the following cases.

\begin{enumerate}
	\item $u = v$. Assume first that $C_i \neq C_j$. Then, since $u \prec_x b$ we have $i < j$ and thus according to Case~\ref{deg-8-case-2}, $u$ is placed as the last vertex of $C_i$ in $\prec_x$. This would imply $a \prec_x u$, which is a contradiction to $(u,a)$ using the $E$-port at $u$. Hence assume that $C_i = C_j$, i.e., $b \in C_i$. Then we are in Case~\ref{deg-8-case-3}. In particular, we placed $u$ such that $u$ and $b$ are consecutive, thus $(u,b)$ is classified as an oblique-$2$ edge, again a contradiction.
	\item $a = v$. Since $(u,a)$ uses the $E$-port at $u$ by assumption, it follows that $u \prec_x a$. By construction, $u$ is the last vertex of $P_i$ in $\prec_x$. But then $a$ was necessarily inserted as the last vertex of $C_i$ in $\prec_x$ such that $u$ and $a$ are consecutive and hence $(u,v)$ is an oblique-$2$ edge.
	
	\item $v$ was inserted directly in between $u$ and $a$. By construction, this only occurs in Case~\ref{deg-8-case-3}.
	In this instance, $v$ was inserted in between $u$ and $a$ and thus either $(v,u) \in F_3$ or $(v,a) \in F_3$. Suppose first the former. By our construction rule, $v$ is placed after $u$ if there exists a vertex $u'$ such that $(u,u') \in F_3$ and $u' \prec_x u$. But this is impossible, as $u'$ and $b$ cannot coincide, since we have $u' \prec_x u \prec_x b$. Hence, $v$ is placed before $a$ and there exists a vertex $a'$ such that $(a,a') \in F_3$ with $a \prec_x a'$. Further, by construction, $a$ is consecutive to one of its neighbors of $F_1$ (the one different from $u$). Hence, we can
	again reorient $e$ such that $e = (a,u)$ uses the $W$-port at $a$ to guarantee \two.

\end{enumerate}

\noindent Observe that if an edge $(u,v)$ was redirected, then both $u$ and $v$ belong to the same cycle $C_i$ of $F_1$ and since this operation has to be performed at most once per cycle, it follows that they can be considered independently. So far, we have computed $\prec_x$ and classified every edge of $F_1 \cup F_3$ guaranteeing Invariants \ref{inv:next-to}-\ref{inv:unique}. Symmetrically, we can compute $\prec_y$ and classify every edge of $F_2 \cup F_4$ guaranteeing the following corresponding versions of Invariants~\ref{inv:next-to}-\ref{inv:unique}:

\begin{enumerate}[label={\bf I.\arabic*}, ref={\arabic*}]\setcounter{enumi}{4}
\item \label{inv:y-next-to} Let $\mathcal{S}$ be a set containing exactly one arbitrary vertex from each of the cycles $\mathcal{C}_1, \mathcal{C}_2,\dots, \mathcal{C}_{\kappa}$ of $F_2$ and denote by $\mathcal{P}_1,\mathcal{P}_2,\dots,\mathcal{P}_{\kappa}$ the resulting paths when restricting the cycles to $V \setminus \mathcal{S}$. Let $u \in \mathcal{S}$ be a vertex of cycle $\mathcal{C}_i$. If $|\mathcal{C}_i| > 1$, then $u$ is placed next to at least one vertex of $\mathcal{P}_i$ in $\prec_y$. Otherwise, $u$ is placed directly after the last vertex of $\mathcal{C}_{i-1}$ (or as first vertex if $i = 1$) in $\prec_y$.
\item \label{inv:y-oblique-2} The endpoints of each vertical oblique-$2$ edge are consecutive in~$\prec_y$.
\item \label{inv:y-oblique-1} Each oblique-$1$ edge $(u,v) \in F_2 \cup  F_4$ is assigned the $S$-port at its source vertex $u$, if $v \prec_y u$; otherwise, it is assigned the $N$-port.
\item \label{inv:y-unique} Every vertical port is assigned at most once.

\end{enumerate}

\subsection{Bend placement}\label{app:bends}
We begin by describing how to place the bends of the edges on each side of the box $B(u)$ of an arbitrary vertex $u$ based on the type of the edge that is incident to $u$, refer to \cref{fig:the-box}. Let $(x_u,y_u)$ be the coordinates of $u$ in $\Gamma$ that are defined by $\prec_x$ and $\prec_y$.
Recall that the box $B(u)$ has size $8 \times 8$. Let $e$ be an edge incident to $u$.
We focus on the case in which $e \in F_1 \cup F_3$, the other case in which $e$ belongs to $F_2 \cup F_4$ is handled symmetrically by simply exchanging $x$ with $y$, ``top/bottom'' with ``right/left'' and ``vertical'' with ``horizontal'' from the following description.
By definition, $e$ is either an oblique-$1$ edge or a vertical oblique-$2$ edge.
Suppose first that $e$ is an oblique-$1$ edge. If $e = (u,v)$, i.e., $e$ is an outgoing edge of $u$ in $F_1 \cup F_3$, then by Invariant~\ref{inv:oblique-1} edge $e$ uses either the $W$- or $E$-port at $u$. In the former case, the segment of $e$ incident to $u$ passes through point $(y_u,y_u-4)$, while in the latter case it passes through point $(y_u,y_u+4)$. For an example, refer to the outgoing edge $(v_3,v_6)$ of $v_3$ in Fig.~\ref{fig:k9}.
If $e = (v,u)$, i.e., $e$ is an incoming edge of $u$ in $F_1 \cup F_3$, then by Invariant~\ref{inv:oblique-1} $e$ uses a horizontal port at $v$ and by the fact that every edge consists of exactly three segments, the vertical segment of $e$ ends at the top or the bottom side of $B(u)$. Since any vertex has at most three incoming edges in $F_1 \cup F_3$ by construction,
we can place the respective bends at $x$-coordinate $x_u+i$ with $i \in \{-2,-1,1,2\}$ and $y$-coordinate $y_u+4$ ($y_u-4$) for the top (bottom) side such that the assigned $i$-value is unique, refer to the incoming edge $(v_4,v_9)$ of $v_9$ in Fig.~\ref{fig:k9}, where $i = -1$.
Finally, the other bend-point of $e$ is uniquely defined as $(x_u+i,y_v)$, since it connects a vertical with a horizontal segment by construction.

Suppose now that $e$ is a vertical oblique-$2$ edge. By \three, $u$ and $v$ are consecutive in $\prec_x$. If $v \prec_x u$ the $x$-coordinate of the bend point is $x_u-4$, otherwise it is $x_u+4$; e.g., refer to the edges $(v_2,v_3)$ and $(v_3,v_4)$ of $v_3$ in Fig.~\ref{fig:k9}, respectively. In order to define the $y$-coordinate of the bend point, we have to consider the relative position of $u$ and $v$ in $\prec_y$. If $v \prec_y u$ the $y$-coordinate of the bend point of $e$ is $y_u-3$ and otherwise it is $y_u+3$. \three implies that any vertex has at most two vertical oblique-$2$ edges since no vertex has more than two direct neighbors in $\prec_x$. From the description of the bend-points, the observation follows:
\begin{obs}\label{obs:no-obique-overlap}
Let $b$ be a bend-point that delimits an oblique segment $s$ which belongs to an edge $e$. If $s$ is incident to $u$, then $b$ does not lie on any other edge incident to $u$.
\end{obs}
\subsection{Proof of correctness}
As a first step we prove that the obtained drawing is in fact \rac{2},
that is, we have to show that Properties~\ref{prp:no-edge-overlap} and \ref{prp:in-the-box} are satisfied. To do so, we introduce the following notation: Let $e = (u,v)$ be an edge of $G$. We define $b_u$ as the first bend-point that is encountered on $e$ starting at $u$, while $b_v$ is the second.
Assume that vertex $u$ is the $i$-th vertex in $\prec_x$ and the $j$-th vertex of $\prec_y$. We define by $\mathcal{H}(u)$ the horizontal strip $[8i-3,8i+3]$ and by $\mathcal{V}(u)$ the vertical strip $[8j-3,8j+3]$; see \cref{fig:the-box}.
The construction in the previous subsection guarantees the following propositions.

\begin{proposition}\label{propo:seg-in-intervall}
Any vertical (horizontal) segment $s$ that belongs to an oblique-$1$ edge $(u,v)$ is contained in $\mathcal{V}(u)$ or $\mathcal{V}(v)$ ($\mathcal{H}(u)$ or $\mathcal{H}(v)$)
\end{proposition}
\begin{proof}
This immediately follows from the description of the bend-points.
\end{proof}

\begin{proposition}\label{propo:no-segment}
Let $u$ be a vertex. Any vertical (horizontal) segment contained in $\mathcal{V}(u)$ ($\mathcal{H}(u)$) belongs to an edge which is incident to $u$.
\end{proposition}
\begin{proof}
    We observe that by construction, no other box $B(v)$ with $v \neq u$ is contained inside $\mathcal{V}(u)$ or $\mathcal{H}(u)$.
    We proof the statement for $\mathcal{V}(u)$, the other case is symmetric. Suppose for a contradiction that the vertical segment $s$ which belongs to $e = (v,w)$ is contained inside $\mathcal{V}_u$, but $v \neq u \neq w$. Suppose first that $e$ is an oblique-$2$ edge. If $e$ is a horizontal oblique-$2$ edge, it does not contain any vertical segment. If $e$ is a vertical oblique-$2$ edge, $v$ and $w$ are consecutive in $\prec_x$ by Invariant~\ref{inv:oblique-2}. In particular, segment $s$ overlaps with the boundary of both $B(v)$ and $B(w)$ and hence cannot be contained in $\mathcal{V}(u)$ since neither $B(v)$ nor $B(w)$ are contained inside $\mathcal{V}(u)$.
    Suppose now that $e$ is an oblique-$1$ edge. Then, segment $s$ is either incident to $v$ or $w$, in which case it is (partially) contained inside $B(v)$ or $B(w)$, or it is a middle segment of $e$ and by construction ends at the boundary of either $B(v)$ or $B(w)$. Since $v \neq u \neq w$ and since no other box besides $B(u)$ is contained inside $\mathcal{V}(u)$, $s$ can not exist.
\end{proof}
In order to show that that Property~\ref{prp:no-edge-overlap} is maintained in $\Gamma$, we consider the following cases:
\begin{enumerate}
    \item $e = (u,v)$ is an oblique-$1$ edge and $b_u$ or $b_v$ lie on another edge
    
     W.l.o.g. assume that $e \in F_1 \cup F_3$ and that $u \prec_x v$. By Invariant~\ref{inv:oblique-1}, this assumption implies that $e$ uses the $E$-port at $u$ and $b_v$ is placed on the top/bottom side of $B(v)$. In particular, $b_v$ is contained inside $\mathcal{V}(v)$,
    and therefore also $b_u$ as they are connected by a vertical segment. Any vertical segment inside $\mathcal{V}(v)$ belongs to an edge that is incident to $v$ by Proposition~\ref{propo:no-segment}. But then Observation~\ref{obs:no-obique-overlap} guarantees no such overlap for $b_v$ and hence for $b_u$ occurs. Symmetrically, a horizontal segment overlapping $b_u$ is inside $\mathcal{H}(u)$  and thus belongs to an edge incident to $u$, in which case we can again use Observation~\ref{obs:no-obique-overlap} to show that no such overlap exists. 
    Finally, consider a horizontal segment that overlaps $b_v$. If this segment belongs to an oblique-$1$ edge $(u',v')$, then it is either contained inside $\mathcal{H}(u')$ or $\mathcal{H}(v')$ - but then $B(u)$ would be contained inside  $\mathcal{H}(u')$ or $\mathcal{H}(v')$, a contradiction. Hence, this segment belongs to an oblique-$2$ edge $(u',v')$, in particular to a horizontal one. The middle segment is contained in exactly two boxes $B(u')$ and $B(v')$ (since $u'$ and $v'$ are consecutive in $\prec_y$ by Invariant~\ref{inv:y-oblique-2}). Hence, no overlap can occur unless $v = u'$ or $v = v'$, but then by Observation~\ref{obs:no-obique-overlap} no overlap occurs.

    \item $e = (u,v)$ is an oblique-$2$ edge and $b_u$ or $b_v$ lies on another edge
    
    By construction, $b_u$ and $b_v$ lie on the boundary of $B(u)$ and $B(v)$, respectively. Since no two boxes overlap and by Observation~\ref{obs:no-obique-overlap}, $b_u$ and $b_v$ do not lie on other oblique segments. Using Proposition~\ref{propo:seg-in-intervall}, we can deduce that $b_u$ and $b_v$ do not lie on a vertical or horizontal segment that belongs to an oblique-$1$ edge, as otherwise $B(u)$ or $B(v)$ would lie in the vertical or horizontal strip of another vertex. Finally, assume that $B(u)$ or $B(v)$ would lie on the middle part of another oblique-$2$ edge $e' = (u',v')$. If $e$ and $e'$ are both vertical/horizontal, this would imply that there exists a pair of vertices with exactly the same $x$-coordinate ($y$-coordinate), which is impossible by construction. Hence, w.l.o.g. assume that $e$ is a vertical oblique-$2$ edge, while $e'$ is a horizontal oblique-$2$ edge and $b_u$ lies on $e'$. But since $b_u$ is placed at $y_u \pm 3$ , $e'$ would be contained in $\mathcal{H}(u)$, a contradiction to Proposition~\ref{propo:no-segment} unless $u'=u$ or $v' = u$, for which Observation~\ref{obs:no-obique-overlap} holds.

\end{enumerate}

Since we established that  Property\ref{prp:no-edge-overlap} holds for $\Gamma$, it remains to show that  Property~\ref{prp:in-the-box} is satisfied.
Consider the first part of Property~\ref{prp:in-the-box}.
Let $s$ be a segment that passes through $B(u)$, but $s$ is not incident to $u$. Since no two boxes overlap and since any oblique segment is contained inside a box, it follows that $s$ can not be oblique, hence it is either vertical or horizontal.
If $s$ belongs to an oblique-$2$ edge $e$, then $s$ is the middle part of $e$, but then by definition $s$ lies on the boundary of $B(u)$ and $B(v)$ and cannot pass through the interior of any box, as otherwise $u$ and $v$ would not be consecutive in $\prec_x$ or $\prec_y$, a contradiction to Invariant~\ref{inv:oblique-2} or Invariant~\ref{inv:y-oblique-2}. If $s$ belongs to an oblique-$1$ edge, then  $s$ is necessarily contained inside the vertical/horizontal strip of one of its endpoints by Proposition~\ref{propo:seg-in-intervall} and thus cannot be contained in the interior of $B(u)$, as otherwise $B(u)$ is contained in the vertical/horizontal strip of a different box, which establishes the first part of~\ref{prp:in-the-box}.

To conclude Property~\ref{prp:in-the-box}, we have to show that any segment that is outside the box is either vertical or horizontal, i.e., that the two end-points that delimit such a segment differ only in $x$- or in $y$-coordinate. To do so, consider any edge $e=(u,v)$.
Suppose first that $e$ is an oblique-$2$ edge.
If $e$ is a vertical oblique-$2$ edge, then $u$ and $v$ are consecutive in $\prec_x$ by \three and $B(u)$ and $B(v)$ are aligned in $x$-coordinate, in particular, there is a vertical line that contains the right side of one box and the left side of the other, hence it passes through the two assigned bend-points, which implies that the middle segment is indeed vertical.
Similarly, we argue if $e$ is a horizontal oblique-$2$ edge.
Suppose now that $e$ is an oblique-$1$ edge. 
\one and Invariant~\ref{inv:y-oblique-1} guarantee that for any relative position of $v$ to $u$, we assigned an appropriate orthogonal port at $u$ which allows to find a point on the first segment, such that the orthogonal middle segment of the edge $e$ (that is perpendicular to the first) can reach the assigned bend point on the boundary of $B(v)$ which shows Property~\ref{prp:in-the-box}.

\begin{remark}
 The resulting drawing of our algorithm is not necessarily simple, i.e., it is possible that two edges have more than one point in common (endpoint and crossing point). While this could be solved in a postprocessing step for vertices that are not part of $S$, it is not immediate how to do it for the set of special vertices. \cite{DBLP:journals/tcs/AngeliniBFK20} showed that restricting the drawing to be simple is in fact a real restriction on the graphs that are realizable in the 1-bend RAC-setting, which intuitively should also hold for the 2-bend case.
\end{remark}

Having proved that the obtained drawing is \rac{2}, we discuss the time complexity and the required area.
We apply \cref{thm:2-factors} to $G$ to obtain $F_1$, $F_2$, $F_3$ and $F_4$ in $\mathcal{O}(n)$ time. For each cycle of $F_1$ and $F_2$, an appropriate ordering of its internal vertices, the classification of the incident edges and the assignment of the orthogonal ports can be computed in time linear in the size of the cycle. Clearly, computing the bend-points can be done in linear time as well. Hence we can conclude that the drawing can be computed in $\mathcal{O}(n)$ time.
For the area, we can observe that the size of the grid defined by the boxes is $8n \times 8n$ and by construction, any vertex and any bend point is placed on a distinct point on the grid.
\end{proof}

\section{Generalization of apRAC}
The following theorem establishes that the class of  \srac{0}{s} graphs forms a proper subclass of the class of $0$-bend RAC graphs when $s \in o(n)$. 

\begin{figure}[h]
    \begin{subfigure}[b]{.32\textwidth}
    \centering
    \includegraphics[scale=1,page=1]{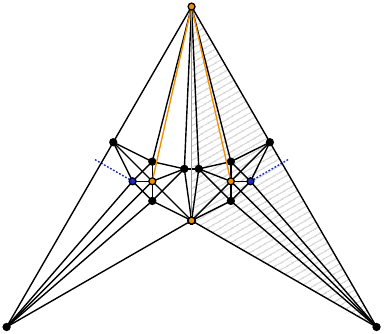}
    \subcaption{}
    \label{fig:inclusion-relation-a}
    \end{subfigure}
    \hfil
    \begin{subfigure}[b]{.65\textwidth}
    \centering
    \includegraphics[scale=1,page=2]{figures/inclusion-relation-violet.pdf}
    \subcaption{}
    \label{fig:inclusion-relation-b}
    \end{subfigure}
    \caption{Illustrations for the proof of \cref{thm:s-rac}.}
    \label{fig:inclusion-relation}   
\end{figure}

\begin{theorem}\label{thm:s-rac}
There exist $0$-bend RAC graphs on $n$ vertices which are not \srac{0}{s} for any 
 $s \in o(n)$.
\end{theorem}

\begin{proof}
Let $G$ be the extension of two augmented square antiprism graphs, which has a unique combinatorial embedding by~\cite{DBLP:journals/jgaa/ArgyriouBS12}; see \cref{fig:inclusion-relation-a}. The final graph $G_k$ consists of $k$ copies of $G$, namely $G_1,\dots,G_k$ such that $G_i$ and $G_{i+1}$ (modulo $k$) are connected by an additional edge (blue in \cref{fig:inclusion-relation-b}) and all copies of $G$ share one vertex $V$ (center in \cref{fig:inclusion-relation-b}). In \cref{fig:inclusion-relation-b} a RAC drawing of $G_6$ is shown. Clearly, this drawing can be extended to any $k \mod 2 = 0$.
Since the blue edges form a horizontal extension the embedding of $G_k$ is unique. But then vertex $v$ is necessarily incident to at least $\frac{3k}{2}$ crossed edges (orange in \cref{fig:inclusion-relation-b}) whose angle formed with the $x$-axis is pairwise different.
Hence, in order to admit a \srac{0}{s} drawing, $s$ has to be at least $\frac{2k}{3}$. But clearly, $k \in \Omega(n)$, a contradiction.
\end{proof}

Regarding \srac{2}{s} graphs, we can extend \cref{thm:ap-2} to provide an upper bound on the edge-density depending on $s$. We remark here that the upper bound of $74.2n$ for general $2$-bend RAC graphs due to \cite{DBLP:journals/comgeo/ArikushiFKMT12} on $n$ vertices can be slightly improved to $\approx 71.9$ by plugging in the new leading constant of the crossing lemma due to \cite{DBLP:journals/comgeo/Ackerman19}.

\begin{theorem}\label{thm:upper-s-rac}
A \srac{2}{s} graph $G$ with $n$ vertices has at most 

$\min\{(6+4s)n-12,71.9n-\mathcal{O}(1)\}$ edges.
\end{theorem}
\begin{proof}
For $s \geq 17$, the upper bound for $2$-bend RAC graphs holds. For every line of unique slope, we can have at most four edge segments incident to every vertex (two parallel and two perpendicular), thus at most
 $4sn$ many edge segments incident to the vertices can be involved in a crossing by definition. Let us denote the edges that contain at least one such segment as $I$ and observe that $|I| \leq 4sn$. Then, $E \setminus I$ consists of edges where only the middle parts can be involved in any crossing. Denote by $D$ the subdrawing of $\Gamma$ restricted to the edges of $E \setminus I$. We claim that the crossing graph of $D$ is bipartite. To see that this yields the desired result, observe that a bipartite crossing graph implies the existence of a two coloring of the edges of $D$ such that no two edges of the same color cross. Hence, the set of edges restricted to a color forms a planar subdrawing and therefore contains at most $3n-6$ edges, thus $|E \setminus I| \leq 6n-12$ and the statement follows.
It remains to proof our claim. By definition, only the middle parts are involved in crossings in $D$. Subdividing the edges at the bends and then restricting $D$ to only contain the edges which corresponded to middle parts in $D$ yields a straight-line RAC drawing $D'$ such that the crossing graph of $D$ and of $D'$ coincide. Since $D'$ is a $0$-bend RAC drawing, we have that the crossing graph of $D'$ is bipartite~\cite{DBLP:journals/tcs/DidimoEL11}, and hence the crossing graph of $D$ is as well, which concludes the proof.
\end{proof}

\section{Conclusion and Open Problems}\label{sec:conclusions}
In this paper, we introduced the class of \rac{k} graphs, gave edge-density bounds, studied inclusion relationships with the general $k$-bend RAC graphs, and concluded with an algorithmic result for graphs with maximum degree $8$.
A natural extension is to allow drawings where each crossing edge-segment is parallel or perpendicular to a line having one out of $s$ different slopes. We denote the class of graphs which admit such a drawing as \srac{k}{s}, and w.l.o.g.\ we assume that the horizontal slope is among the $s$ ones. Observe that for $s=1$, the derived class coincides with the class of \rac{k} graphs.
By joining several copies of the graph supporting \cref{prop:antiprism-two-embbeddings} that all share a common vertex, we show that \srac{0}{s} graphs form a proper subclass of $0$-bend RAC graphs for any $s \in o(n)$; see \cref{thm:s-rac}, which generalizes \cref{cor:separatating}. 
We also adjust the proof of \cref{thm:ap-2} to derive an upper bound on the edge density of \srac{2}{s} graphs; see \cref{thm:upper-s-rac} and observe that our upper bound is better than the one of~\cite{DBLP:journals/comgeo/ArikushiFKMT12} that holds for general $2$-bend RAC graphs for values of $s$ up to $17$. We conclude with the following open problems.
\begin{itemize}
    \item Are there $2$-bend RAC graphs that are not \rac{2}?
    \item 
    For $k \in \{1,2\}$, our edge-density bounds do not relate to the simplicity of the drawings. 
    Are bounds different for simple drawings, as in the general $1$-bend RAC case~\cite{DBLP:journals/tcs/AngeliniBFK20}?

    \item For $k \in \{1,2\}$, does the class of \srac{k}{s} graphs on $n$ vertices coincide with the corresponding class of $k$-bend RAC graphs, when $s \in o(n)$?
    
     \item Based on our exploration of \rac{2} graphs, we conjecture that the edge density of general $2$-bend RAC graphs on $n$ vertices is $10n - \mathcal{O}(1)$.
     
     \item Another important open problem in the field is to settle the complexity of the recognition of general 1-bend RAC graphs. What if we restrict ourselves to the axis-parallel setting?
     
     \item From a practical perspective, can an advantage in the readability of \rac{k} drawings over $k$-bend RAC drawings be demonstrated in user studies?
 \end{itemize}

\newpage
\bibliography{bibliography}

\end{document}